\documentclass[article, a4paper, twoside]{memoir}
\pdfoutput=1 

\counterwithout{section}{chapter}

\setsecnumdepth{subsection}

\setlrmarginsandblock{3.5cm}{2.5cm}{*}
\setulmarginsandblock{2.5cm}{*}{1}
\checkandfixthelayout

\usepackage[american]{babel}
\usepackage{csquotes}

\usepackage{amsmath,amssymb,amsthm}
\usepackage{bm,bbm}
\usepackage{mathtools}

\newtheorem{prop}{Proposition}

\usepackage[dvipsnames]{xcolor}

\usepackage{graphicx}

\usepackage{rotating}

\newcommand{\expect}[1]{ \mathbb{E} \left[ #1 \right] }

\newcommand{\cexpect}[2]{ \mathbb{E} \left[ #1 \: \big \vert #2 \: \right] }

\newcommand{\expectunder}[2]{ \mathbb{E}_{#1} \left[ #2 \right] }

\newcommand{\proba}[1]{ \mathbb{P} \left( #1 \right) }

\newcommand{\cproba}[2]{ \mathbb{P} \left( #1 \: \big \vert \: #2 \right) }

\newcommand{\vari}[1]{ \mathbb{V} \text{ar} \left( #1 \right) }

\newcommand{\covari}[2]{ \mathbb{C} \text{ov} \left( #1 , #2 \right) }

\newcommand{\correl}[2]{ \mathbb{C} \text{orr} \left( #1 , #2 \right) }

\newcommand{\setdef}[1]{ \left\{ #1 \right\} }

\newcommand{\indic}[1]{ \mathbbm{1}_{ \left\{ #1 \right\} } }

\newcommand{\absval}[1]{ \lvert #1 \rvert }

\DeclareMathOperator{\logit}{logit}

\newcommand{\confint}[2]{$95\%$ CI: $\left[ #1; #2 \right]$}

\usepackage[round]{natbib}
\usepackage{hyperref}

\hypersetup{
  colorlinks = true,
  linkcolor = black, citecolor = black, filecolor = black, urlcolor = blue,
  pdftitle = {A bivariate logistic regression model based on latent variables},
  pdfauthor = {Simon Bang Kristensen, Bo Martin Bibby}
}

\title{A bivariate logistic regression model based on\\ latent variables}

\author{Simon Bang Kristensen\thanks{Correspondance should be sent to Simon Bang Kristensen, Department of Public Health, Bartholins All\'{e} 2, 8000 Aarhus C, DK-Denmark, \href{mailto:simonbk@ph.au.dk}{simonbk@ph.au.dk}.} \& Bo Martin Bibby \\
\normalsize Department of Public Health, Biostatistics, Aarhus University}

\date{
  \today
}

\begin{document}

\maketitle

\begin{abstract}
\noindent
Bivariate observations of binary and ordinal data arise frequently and require a bivariate modelling approach in cases where one is interested in aspects of the marginal distributions as separate outcomes along with the association between the two. We consider methods for constructing such bivariate models with logistic marginals and propose a model based on the Ali-Mikhail-Haq bivariate logistic distribution. We motivate the model as an extension of that based on the Gumbel type 2 distribution as considered by other authors and as a bivariate extension of the logistic distribution which preserves certain natural characteristics. Basic properties of the obtained model are studied and the proposed methods are illustrated through analysis of two data sets, one describing the trekking habits of Norwegian hikers, the other stemming from a cognitive experiment of visual recognition and awareness.
\end{abstract}

\section{Introduction}
\label{sec:introduction}

In the present paper we consider the problem of modelling observations $(X,Y)$ where $X$ is binary and $Y$ is ordinal on the scale $\setdef{1, \ldots, K}$ for some $K \geq 2$. Such observations are collected for subjects $i=1, \ldots, N$ and in some cases longitudinally at time points $t=1, \ldots, T$. We denote a single instance as $(X_{it}, Y_{it})$ dropping the subscript $t$ when only one observation is collected for each individual.

Data of this type occur for example in clinical trials, epidemiology and cognitive science. We focus on scenarios where one is interested both in $X$ and $Y$ as outcomes in themselves but also in their association. For example in an efficacy-toxicity study, interest may centre on both describing the efficacy and toxicity of the drug but also in modelling for example the joint probability of benefiting from a positive treatment outcome ($X=1$) while not experiencing a high degree of side-effects ($Y \leq k$ for some $k$) as a function of dosage. Another example arises in the cognitive sciences where a participant in a trial must perform a task (for example, identify a geometric shape presented on a screen) and quantify her/his confidence in the assertion (say, on an ordinal rating scale). The trial is then repeated under various configurations, for example manipulating the task difficulty or stimulus intensity (e.g. the time that the image is shown). Often one is interested in the accuracy of the performance ($X$) and the quantifications of confidence ($Y$) as a function of the stimulus intensity frequently occurring as a sigmoidal function, but one could also wish to quantify for instance the participants' ability to discriminate correct from incorrect responses. The latter would be an expression of metacognitive ability, meaning a person's ability to reflect on and evaluate decisions and performances \citep[e.g.][]{metcalfe_metacognition:_1996}, metacognitive deficiency being a characteristic of diagnoses such as schizophrenia, Alzheimer's disease and certain injuries to the brain \citep{david_failures_2012}. We return to such an experiment later. 

The type of models that are the main topic of the present paper relies on latent variables $X^{\ast}$ and $Y^{\ast}$ and we will generally specify a model with two components: A rule that links the behaviour of the latent variables to the observed (the threshold model), and a bivariate distribution for the latent variables (the distributional model). We will throughout write $F$ and $G$ for the cumulative distribution functions for $X^{\ast}$ and $Y^{\ast}$, respectively, and $H$ for the joint distribution of $(X^{\ast}, Y^{\ast})$. 

Before introducing the main model of the present paper, we first review and delimit methods for constructing bivariate logistic models.

\subsection{Bivariate logistic models}
\label{sec:bivar-logist-models}

An initial, natural model for $(X,Y)$ may be obtained by considering first the marginal models separately. Marginally, one might be interested in applying a logistic regression model for $X$ and a proportional odds model for $Y$. Thus, we might pose the following two regression equations given covariates $Z_1$ and $Z_2$,
\begin{equation}
  \label{eq:4}
  \begin{aligned}
    &\logit \proba{X_i = 1} = \theta + Z_{1,i} \beta_x, \\
    &\logit \proba{Y_i \leq k} = \tau_k - Z_{2,i} \beta_y .
  \end{aligned}
\end{equation}
Note that the sign for $\beta_y$ is chosen so that its interpretation coincides with that in standard logistic regression for the case where $Y$ is binary ($K=2$). An important part of the model for $Y$ is the so-called proportional odds assumption \citep{mccullagh_regression_1980-1}, which assumes that the covariates do not interact with the scale level of the ordinal variable, or that the effect of the covariates affecting the distribution of the latent variable for $Y$ is independent of $k$. This implies that the marginal log-odds ratios $\beta_y$ are independent of $k$. Both models in (\ref{eq:4}) can be written as a latent variable model. Defining threshold models $\setdef{X=1} = \setdef{X^{\ast} > \theta}$ and $\setdef{Y \leq k} = \setdef{Y^{\ast} \leq \tau_k}$ for $K$ threshold parameters $\theta$ and $\tau_1 < \ldots < \tau_{K-1}$, we may stack the latent observations to obtain the following distributional model,
\begin{equation}
  \label{eq:2}
  \begin{pmatrix}
    X^{\ast}_{i} \\
    Y^{\ast}_{i}
  \end{pmatrix} =
  \begin{pmatrix}
    Z_{1,i} & 0 \\
    0 & Z_{2,i}
  \end{pmatrix}
  \begin{pmatrix}
    \beta_x \\
    \beta_y
  \end{pmatrix} +
  \begin{pmatrix}
    \epsilon_{x,i} \\
    \epsilon_{y,i}
  \end{pmatrix} ,
\end{equation}
where $\epsilon_{x,i}$ is independent of $\epsilon_{y,i}$ both having standard logistic distributions. It is simple to check that the two marginal models defined by (\ref{eq:2}) are equivalent to those in (\ref{eq:4}).
The obvious generalization of (\ref{eq:2}) is to model $\epsilon_x$ and $\epsilon_y$ as being dependent, as we will do in the present paper. We will return to expansions including random effects for the longitudinal setting in the discussion. The choice of a bivariate model for $(\epsilon_x, \epsilon_y)$ is, however, not straightforward as various bivariate logistic distributions exist.

A classic paper, \cite{gumbel_bivariate_1961}, introduces a bivariate logistic distribution where the joint cumulative distribution function $H$ is given by,
\begin{equation}
  \label{eq:5}
  (u, v) \mapsto \left(1 + e^{-u} + e^{-v}\right)^{-1} , \quad
  u, v \in \mathbb{R} ,
\end{equation}
often called the Gumbel type 1 distribution. The Gumbel type 1 distribution was extended by \cite{satterthwaite_generalisation_1978} but to a model with generalised logistic rather than logistic marginals. \citeauthor{gumbel_bivariate_1961} also proposed a model $H$ with an association parameter $\omega \in [-1 , 1]$, a distribution of the Farlie-Gumbel-Morgernstern type called the Gumbel type 2 distribution,
\begin{equation}
  \label{eq:6}
  (u,v) \mapsto F(u) G(v) \left[ 1 + \omega (1-F(u))(1-G(v)) \right] , \quad u,v \in \mathbb{R} .
\end{equation}
\cite{murtaugh_bivariate_1990} studied a latent variable, bivariate logistic model using the type 2 distribution in (\ref{eq:6}), a model investigated further in \cite{heise_optimal_1996}. Their motivation was efficiency-toxicity studies such as those described above and their approach focused on the benefit of the bivariate model with regard to statistical efficiency and for estimating the joint probability of non-toxic effective treatment. 
\cite{malik_multivariate_1973} extended the Gumbel type 2 distribution to higher dimensions, a model used by \cite{li_two-dimensional_2011} for a multivariate logistic model for multiple outcomes in which they also studied association measures.

\subsubsection{Measures of association}
\label{sec:measures-association}

When discussing measures of the association between $X$ and $Y$, it may be practical to discern between two classes of methods for constructing bivariate logistic models. In one, the marginals are specified together with a predetermined association measure to yield the bivariate probabilities. In the other, the bivariate distribution is specified which then leads to the association measure. The former may be said to stem from an assumption of no effect modification on a specific scale, while the latter may be motivated by some property of the bivariate distribution.  
A prevalent example of the former arises from the so-called Plackett construction. \cite{plackett_class_1965} considers the problem of constructing a (not necessarily discrete) bivariate distribution by considering an association measure $\psi$ given by,
\begin{equation}
  \label{eq:1}
  \psi (u, v) = \frac{
    H(u,v) \left(1 - F(u) - G(v) + H(u,v)\right)
  }{
    \left(F(u) - H(u,v)\right) \left(G(v) - H(u,v)\right)
  } .
\end{equation}
When $\psi$ is a known function and for a given pair of marginals $F$ and $G$, (\ref{eq:1}) may be solved as a quadratic equation in $H$ and thus may by seen as a defining equation for a bivariate distribution \citep{mardia_contributions_1967}. Under the particular assumption that $\psi$ is constant, $\psi(u, v) \equiv \psi > 0$, the approach defines a two-dimensional model. The resulting bivariate distribution $H$ is often called the Plackett distribution or, when applied as latent variables for discrete data, the constant global odds ratio model \citep{molenberghs_models_2005}, which has been used to model multivariate ordinal data \citep{molenberghs_marginal_1994, forcina_regression_2008}. Note that when this distribution is subsequently discretised at $(u, v)$ and used as latent variables for a bivariate discrete distribution, $\psi$ is the odds ratio in this distribution which may be the motivation for \cite{mardia_contributions_1967} to refer to the distribution as ``contingency-type''. \cite{lipsitz_maximum_1990} considers models for bivariate dichotome variables arising from the specification of marginals along with an association measure, i.e. the correlation (leading to the so-called Bahadur model), the odds ratio (the Plackett model) as well as the relative risk.

In the bivariate logistic regression model that uses the Plackett distribution for the latent variables, it is not only the marginal odds ratios for $Y$ that are independent of the scale level $k$ but the odds ratio $\psi$ which describes the association between $X$ and $Y$ will also not depend on $k$. This is a consequence of the fact that $\psi$ in this case does not depend on the marginal models meaning also that the effect of $X$ on $Y$ is not modified by any covariates from the marginal distributions. However, it has been proposed to introduce covariates in a model for $\log \psi$ \citep{molenberghs_marginal_1994, forcina_regression_2008}. This property is in contrast to models defined from a bivariate latent distribution such as the Gumbel type 2 distribution in (\ref{eq:6}). For such a model, $\psi$ will usually depend on both marginal models. Thus, every covariate from the marginals may be viewed as an effect modifier for the relationship between $X$ and $Y$ and additionally, the association measure will usually depend on $k$.

\bigskip \noindent
The remainder of the paper is structured as follows. We begin by introducing the primary model of the paper, that based on the Ali-Mikhail-Haq (AMH) construction \citep{ali_class_1978}, and study basic properties and the likelihood function of the model. We further study some theoretical properties of association measures in the AMH bivariate logistic model. We illustrate these ideas by applying them to a classic data set on Norwegian hiking patterns and to a cognitive experiment of the form described above. Finally, we discuss methods to include random effects in the proposed model.

\section{The AMH bivariate logistic model}
\label{sec:amh-bivar-logis}

Below, we introduce the AMH bivariate logistic distribution and the bivariate logistic regression model arising from using the AMH distribution as latent variables. We motivate this choice of model in two ways: As an extension of the bivariate logistic regression using Gumbel type 2 latent variables as studied by other authors \citep{murtaugh_bivariate_1990, heise_optimal_1996, li_two-dimensional_2011} and as a ``natural'' two dimensional logistic distribution preserving a salient property of the one-dimensional logistic distribution.

\cite{ali_class_1978} propose to generalise the univariate logistic distribution by searching for bivariate cumulative distribution functions given by
\begin{equation}
  \label{eq:7}
  H(u, v) = \frac{1}{1 + A(u, v)} ,
\end{equation}
where $A$ is called a bivariate odds function extending the univariate notion of an odds. Let $A_x$ and $A_y$ be the respective marginal odds functions, i.e. $A_x(u) = \proba{X > u}$. Motivated by the definition of a class of distributions containing both the case of independence, $A=A_x + A_y + A_x A_y$, as well as Gumbel's type 1 distribution in (\ref{eq:5}) corresponding to $A=A_x + A_y$, the authors propose to search for an odds function satisfying the differential equation,
\begin{equation}
  \label{eq:8}
  \frac{\partial^2}{\partial A_x \partial A_y} A = 1-\omega ,
\end{equation}
for $\omega \in [-1,1]$. They argue that the unique solution (such that $H$ is a bivariate distribution) is $A=A_x + A_y + (1-\omega) A_x A_y$, which particularly by choosing logistic odds functions $A_x(u)=e^{- u}=A_y(u)$ leads to the AMH bivariate logistic distributions,
\begin{equation}
  \label{eq:9}
  H(u, v) = \frac{1}{1 + e^{-u} + e^{-v} + (1-\omega) e^{-u-v}} .
\end{equation}
We denote by $\Lambda_2^{\ast}(\omega)$ this class of distributions and define $\Lambda_2^{\ast} (\mu, \nu; \omega)$ to be the AMH distribution with locations $\mu$, $\nu$, i.e. having distribution function $H(u-\mu, v-\nu)$. Contour plots of the density $\partial^2 H(u, v) / \partial u \partial v$ are shown in Figure \ref{fig:contours}. Note that the marginals are both standard logistic with $F=G$ not depending on $\omega$, the marginal distributions are independent when $\omega=0$ and the case $\omega=1$ corresponds to Gumbel's type 1 distribution in (\ref{eq:5}).

\begin{figure}[pht]
  \centering
  \makebox[\textwidth][c]{\includegraphics[width=1.35\linewidth]{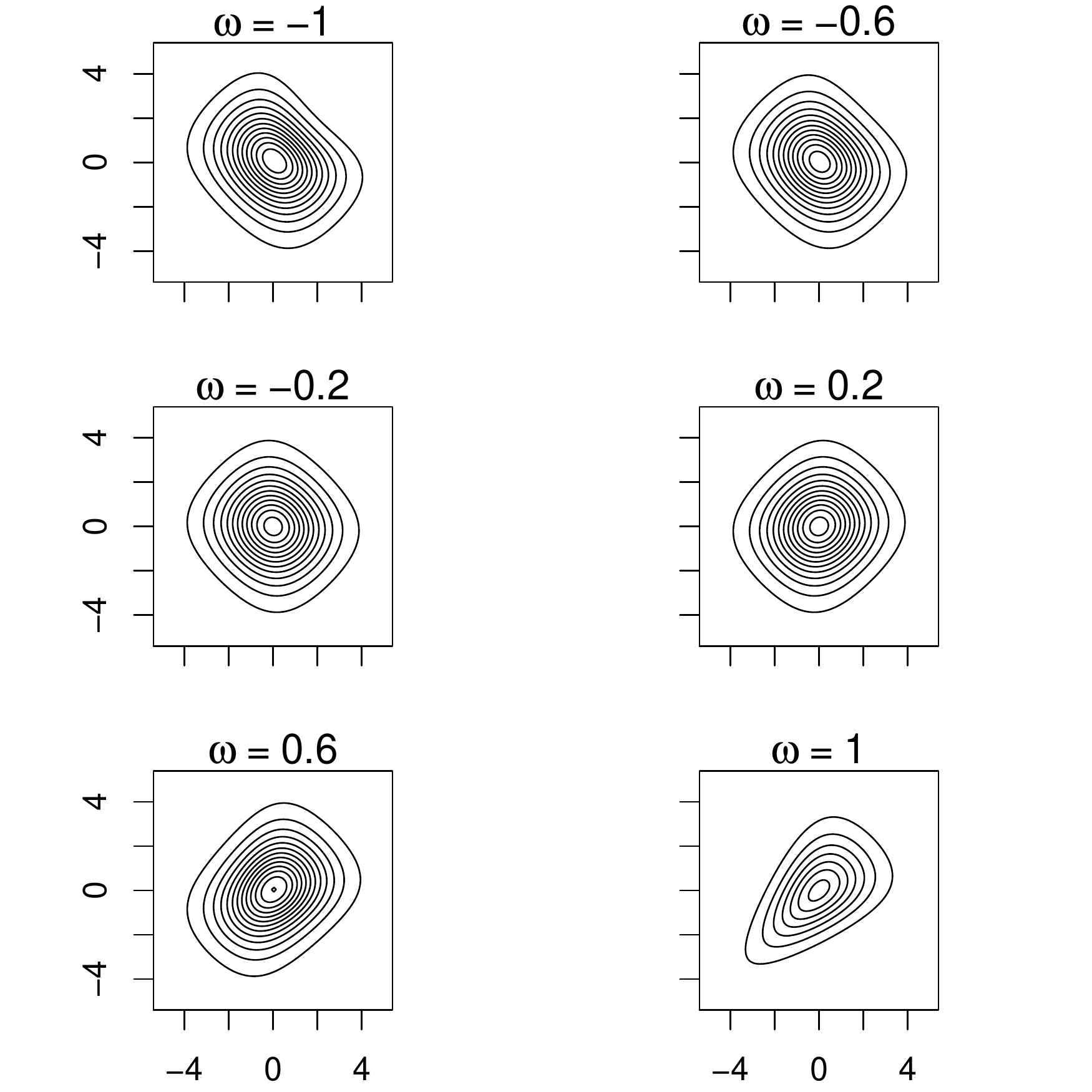}}
  \caption{Contour plots of the AMH distribution $\Lambda_2^{\ast}(\omega)$ with logistic marginals in (\ref{eq:9}) for varying values of the association parameter $\omega$.}
  \label{fig:contours}
\end{figure}

A connection between the AMH-distribution and Gumbel's type 2 distribution in (\ref{eq:6}) can be illustrated by a series expansion of $H$,
\begin{equation}
  \label{eq:18}
  H(u, v) =
  F(u) G(v)
  \sum_{n=0}^{\infty} \omega^n \left[1 - F(u)\right]^n \left[1 - G(v)\right]^n ,
\end{equation}
which shows that we may interpret the Gumbel type 2 distribution as a first order approximation to the AMH distribution \citep[][Chapter 3, ex. 3.40]{nelsen_introduction_2006}.

The following result provides a more probabilistic motivation (also see the discussion by Arnold in \cite{balakrishnan_handbook_1991}, Chapter 11). Recall that the logistic distribution is closed under geometric maximisation (and minimisation) in the sense that if $\setdef{u_i}$ is an i.i.d. sequence of standard logistic variables and $M \sim \text{geom}(\pi)$ is an independent geometric variable, then $\max_{i \leq M} u_i$ is again logistic with location $-\log(\pi)$. The following proposition shows that the AMH distribution possesses a two-dimensional version of this property for positive $\omega$.

\begin{prop}[AMH as a geometric mixture]\label{sec:prop-amh-mixture}
  Let $\omega \in (0,1)$ and consider the distribution $H$ in (\ref{eq:9}). Define $z_1 = u - \log(1-\omega)$ and $z_2 = v - \log(1-\omega)$. Then,
  \begin{equation}
    \label{eq:19}
    H(u, v) = \expectunder{M \sim \text{geom} (1-\omega)}{F^M (z_1) F^M (z_2)} .
  \end{equation}
\end{prop}
\begin{proof}
  $H$ is easily rewritten to resemble the probability generating function of a geometric variable evaluated at $F (z_1) F (z_2)$, which is exactly the right hand side of (\ref{eq:19}), see Appendix \ref{sec:proof-proposition} for details.
\end{proof}

We may thus represent the AMH distribution as follows. Draw three independent replicates $M_x$, $M_y$ and $M$ from the distribution $\text{geom} (1-\omega)$, and let $\setdef{u_i}$ and $\setdef{v_j}$ be two i.i.d. sequences of logistic variables representing the marginals. We draw  independent samples of stochastic sizes $M_x$ and $M_y$ from the two marginal populations and repeat this $M$ times to obtain the matrices $\setdef{u_{ij}}_{i=1, \ldots, M, j=1, \ldots, M_x}$ and $\setdef{v_{ij}}_{i=1, \ldots, M, j=1, \ldots, M_y}$. Now if we set,
\begin{equation}
  \label{eq:28}
  X^{\ast} = \max_{i=1, \ldots, M} \left( \min_{j=1, \ldots, M_x} u_{ij} \right) ,
  \quad \text{ and } \quad
  Y^{\ast} = \max_{i=1, \ldots, M} \left( \min_{j=1, \ldots, M_y} v_{ij} \right) ,
\end{equation}
then it follows from the proposition that $(X^{\ast}, Y^{\ast})$ follows a AMH distribution with association parameter $\omega$, since $\min_{j=1, \ldots, M_x} u_{j}$ and $\min_{j=1, \ldots, M_y} v_{j}$ have respective distribution functions $F(z_1)$ and $F(z_2)$. Intuitively, the dependence comes from the fact that the maximum is an increasing function: If we draw a small $M$, both of the maxima are based on a small number of repetitions and we would expect both to be smaller than if we had drawn a large $M$. Equivalently, a small/large observation of $X^{\ast}$ would lead us to suspect that $M$ is small/large, and we would consequently expect a small/large observation of $Y^{\ast}$ -- thus the two are positively correlated.

We will use the AMH bivariate logistic distribution $\Lambda_2^{\ast}(\omega)$ as latent variables for two-dimensional observations $(X,Y)$ where $X$ is Bernoulli and $Y$ is ordinal on the scale $\setdef{1,\ldots,K}$. Let $X^{\ast}$ and $Y^{\ast}$ be two latent variables with joint distribution $\Lambda_2^{\ast} (\omega)$ and assume the following threshold model,
\begin{equation}
  \label{eq:27}
  \begin{aligned}
    &\setdef{X = 1} = \setdef{X^{\ast} > \theta} \\
    &\setdef{Y \leq k} = \setdef{Y^{\ast} \leq \tau_k} , \quad k=1, \ldots, K-1
  \end{aligned}
\end{equation}
where $(\theta, \tau_1, \ldots, \tau_{K-1})^T$ is a set of threshold parameters with \mbox{$\tau_1 < \ldots < \tau_{K-1}$}. Note in particular that we have assumed that the covariate effects on the distribution of $Y^{\ast}$ do not depend on $k$ implying a proportional odds assumption in the sense discussed in the introduction
. We refer to the resulting discrete distribution as the bivariate logistic regression model with AMH latent variables and denote this $\Lambda_2 (\omega)$. When there are covariates these are included in the location of the latent variables and we denote the resulting model $\Lambda_2 (Z_1 \beta_x, Z_2 \beta_y ; \omega)$, having the regression models in (\ref{eq:4}) as marginals. The probability function for this discrete distribution is shown in Figure \ref{fig:observ-hist}.

\begin{figure}[pht]
  \centering
  \makebox[\textwidth][c]{\includegraphics[width=1.35\linewidth]{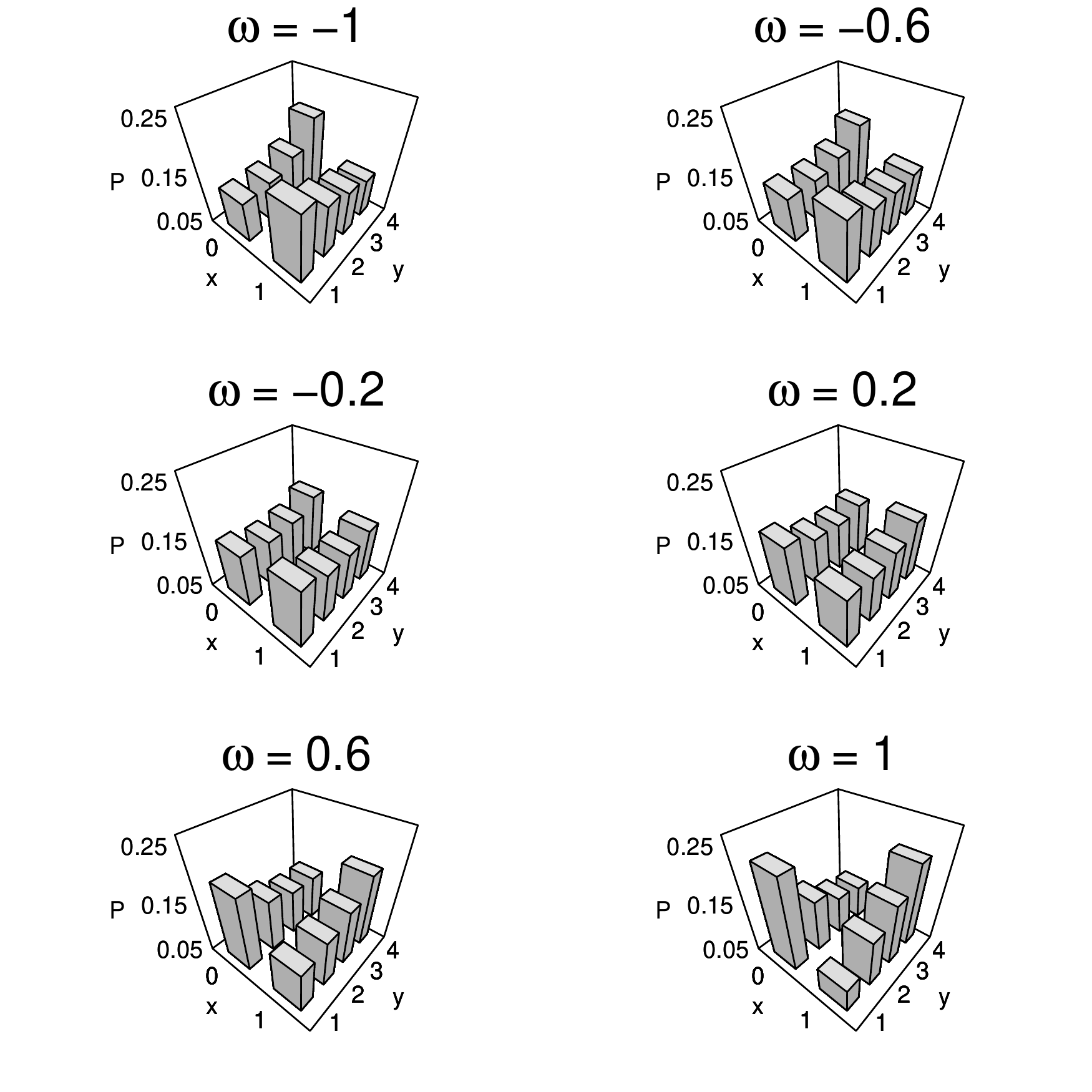}}
  \caption{Histograms for the observed distribution $\Lambda_2 (\omega)$ with $K=4$ for varying values of the association parameter $\omega$. The thresholds were set to $\theta = 0$, $\tau_1 = -1$, $\tau_2 = 0$ and $\tau_3 = 1$.}
  \label{fig:observ-hist}
\end{figure}

\section{Basic properties}
\label{sec:basic-properties}

Below we derive probabilities, moment generating functions and moments for the models $\Lambda_2^{\ast} (\omega)$ and $\Lambda_2 (\omega)$, where we suppress the dependence on covariates as the results may easily be extended by standard results for location families. The first and second moments of the latent and observed variables are summarised in Table \ref{tab:amh-moments}. The covariance for the latent variables is most easily expressed using the series expansion in (\ref{eq:18}), see \cite{mikhail_regression_1987} for some similar results. When $\omega = 1$ we recognise the Basel series so that the covariance is $\pi^2/6$ meaning that the correlation is $1/2$ in concordance with \cite{gumbel_bivariate_1961}. Note that this constitutes the upper bound on the correlation, while the lower bound for $\omega=-1$ is $-1/4$. Thus, the latent variables and the derived model for the observed distribution are more naturally applied to model positively correlated variables. 
The Spearman correlation, which we do no give here, between the latent variables is given in \cite[][Chapter 5, ex. 5.10]{nelsen_introduction_2006} in terms of the dilogarithm, but can also  be shown to admit an appealing representation as a power series in $\omega$.

\begin{sidewaystable}[ph]
  \centering
  \begin{tabularx}{0.75\linewidth}{l | c | c || c | c}
    & $X^{\ast}$ & $Y^{\ast}$ & $X$ & $Y$ \\
    \hline 
    \rule{0pt}{2em}$\expect{\cdot}$ & 0 & 0 &
                                              $\frac{e^{-\theta}}{1 + e^{-\theta}}$ & $\sum_{k=1}^{K} k p_Y(k) =
                                                                                      \sum_{k=1}^{K} k
                                                                                      \frac{
                                                                                      e^{-\tau_{k-1}} - e^{-\tau_{k}}
                                                                                      }{
                                                                                      \left(1 + e^{-\tau_{k-1}}\right) \left(1 + e^{-\tau_{k}}\right)
                                                                                      }$ \\[0.45cm]
    $\vari{\cdot}$ &
                     $\frac{\pi^2}{3}$ &
                                         $\frac{\pi^2}{3}$ &
                                                             $\frac{e^{-\theta}}{\left( 1 + e^{-\theta} \right)^2}$ &
                                                                                                                      $\begin{aligned}
                                                                                                                        &p_{Y}(1) (1-p_{Y}(1))   \\
                                                                                                                        &\quad +\sum_{k=2}^{K}k \sum_{j=1}^{k-1}
                                                                                                                        \frac{k p_Y(k) (1-p_Y(k)) - 2 j (k-1) p_Y(j) p_Y(k)}{k-1}
                                                                                                                      \end{aligned}$ \rule[-3em]{0pt}{0pt}\\
    \hline
    \rule{0pt}{4em}$\covari{\cdot}{\cdot}$ &
                                             \multicolumn{2}{c||}{$\sum_{n=1}^{\infty} \frac{\omega^n}{n^2}$}&
                                                                                                               \multicolumn{2}{c}{$
                                                                                                               \begin{aligned}
                                                                                                                 \sum_{k=1}^{K} k
                                                                                                                 \left\{
                                                                                                                   F(\theta) \left[F(\tau_k) - F(\tau_{k-1})\right] -
                                                                                                                   \left[H(\theta, \tau_k) - H(\theta, \tau_{k-1})\right]
                                                                                                                 \right\} \\
                                                                                                                 \overset{K=2}{=}
                                                                                                                 \frac{
                                                                                                                   \omega \cdot e^{-\theta-\tau_k}
                                                                                                                 }{
                                                                                                                   \left[1 + e^{-\theta} + e^{-\tau_k} + (1-\omega)e^{-\theta-\tau_k} \right]
                                                                                                                   \left[1 + e^{-\theta} \right]
                                                                                                                   \left[1 + e^{-\tau_k} \right]
                                                                                                                 }
                                                                                                               \end{aligned}
    $}
  \end{tabularx}
  \caption{\rule{0pt}{2em}First and second moments of the latent AMH distribution and the observed variables. The covariance between the observed variables is expressed in terms of probability functions for general $K$ and more explicitly in terms of the parameters for $K=2$. Note that $p_Y(y) = \proba{Y=y}$ is the probability mass function for $Y$.}
  \label{tab:amh-moments}
\end{sidewaystable}

The regression of the latent variables onto each other may be derived from the conditional moment generating function of $Y^{\ast}$ given $X^{\ast}$, which is of the form,
\begin{equation}
  \label{eq:13}
  \begin{aligned}
    \mathcal{L}_{Y^{\ast} \: \vert \: X^{\ast}} \left( t \: \vert \: \theta \right)
    &= \cexpect{e^{t Y^{\ast}}}{X^{\ast} = \theta} \\
    &=
    \frac{
      \Gamma(t+2) \Gamma(1-t)
    }{
      2 (1+e^{-\theta})^t (1+(1-\omega) e^{-\theta})^{1-t}
    } \\
    &\quad \cdot
    \left\{
    1+\omega +
    (1-\omega) \left\{
      e^{-\theta} + \left[1 + e^{-\theta} \right] \frac{1-t}{1+t}
    \right\}
    \right\} .
  \end{aligned} 
\end{equation}
By differentiation of the log conditional moment generating function we see in particular that,
\begin{equation}
  \label{eq:15}
  \begin{aligned}
    \cexpect{Y^{\ast}}{X^{\ast} = \theta} &=
    1 + \log \left(
      \frac{1+(1-\omega)e^{-\theta}}{1+e^{-\theta}}
    \right) \\
    &\quad -
    \frac{
      2 (1-\omega) \left[ 1 + e^{-\theta} \right]
    }{
      1+\omega   +
      (1-\omega) \left[
        1 +
        2 e^{-\theta}
      \right]
    } .
  \end{aligned}
\end{equation}

For notational convenience we introduce the indicators $y^k = \indic{y = k}$ and set $\tau_0 = - \infty$ and $\tau_K = \infty$. The probability mass function of the observed distribution is then given by,
\begin{equation}
  \label{eq:24}
  \begin{aligned}
    &p(x, y) = \proba{X=x, Y=y} \\
    &= \prod_{k=1}^{K} \proba{X = 0, Y=k}^{(1- x) y^k} \proba{X = 1, Y=k}^{x y^k} \\
    &=
    \prod_{k=1}^{K} \left[
      \proba{X = 0, Y \leq k} - \proba{X = 0, Y \leq k-1}
    \right]^{(1- x) y^k} \\
    &\quad \cdot
    \left[
      \proba{X = 0, Y \leq k-1} - \proba{Y \leq k-1} -
      \left(
        \proba{X = 0, Y \leq k} - \proba{Y \leq k}
      \right)
    \right]^{x y^k} \\
    &=
    \left[
      H(\theta, \sum_{k=1}^{K} y^k \tau_k) -
      H(\theta, \sum_{k=1}^{K} y^k \tau_{k-1})
    \right]^{(1- x)} \\
    &\quad \cdot
    \Bigg[
      H(\theta, \sum_{k=1}^{K} y^k \tau_{k-1}) -
      F(\sum_{k=1}^{K} y^k \tau_{k-1}) \\
      &\qquad -
      \left(
        H(\theta, \sum_{k=1}^{K} y^k \tau_k) -
        F(\sum_{k=1}^{K} y^k \tau_k)
      \right)
    \Bigg]^{x},
  \end{aligned}
\end{equation}
for $x=0,1$ and $y=1,\ldots, K$.

\section{Estimation}
\label{sec:likelihood-function}

We may perform estimation in the bivariate AMH logistic regression model by maximum likelihood. The model is parametrised by
\begin{equation}
  \label{eq:130}
  \psi = (\theta, \tau_1, \ldots, \tau_{K-1}, (\beta_x)^T, (\beta_y)^T, \zeta(\omega))^T ,
\end{equation}
where $\zeta(s) = \text{tanh}^{-1}(s)$ is the Fisher $\zeta$-transform, having the advantage that $\zeta$ varies unconstrained on the real line. Suppose that we for a subject $i$ observe data $(x_i,y_i)$. Define, as above the indicators $y_i^k = \indic{y_i = k}$ as well as $\tau_0 = - \infty$ and $\tau_K = \infty$. Then using the pmf in (\ref{eq:24}), the log-likelihood is given as,
\begin{equation}
  \label{eq:129}
  \begin{aligned}
    l (\psi) &= \sum_{n=1}^{N} \Bigg\{
    (1-x_i)
    \log
    \left[
      H(\theta-Z_1 \beta_x, \sum_{k=1}^{K} y_i^k \tau_k - Z_2 \beta_y) -
      H(\theta-Z_1 \beta_x, \sum_{k=1}^{K} y_i^k \tau_{k-1} - Z_2 \beta_y)
    \right] \\
    &\quad +
    x_i
    \log
    \Bigg[
    H(\theta-Z_1 \beta_x, \sum_{k=1}^{K} y_i^k \tau_{k-1} - Z_2 \beta_y) -
    F(\sum_{k=1}^{K} y_i^k \tau_{k-1} - Z_2 \beta_y) \\
    &\qquad -
    \left(
      H(\theta-Z_1 \beta_x, \sum_{k=1}^{K} y_i^k \tau_k - Z_2 \beta_y) -
      F(\sum_{k=1}^{K} y_i^k \tau_k - Z_2 \beta_y)
    \right)
    \Bigg]
    \Bigg\}
  \end{aligned}
\end{equation}
Note that $H$ in (\ref{eq:129}) is the AMH bivariate logistic distribution from (\ref{eq:9}) with association parameter $\omega = \text{tanh}(\zeta)$.

Parameter estimates are obtained by maximisation of (\ref{eq:129}), for example using numerical, gradient-based methods under the constraint $\tau_1 < \tau_2 < \ldots < \tau_{K-1}$. Note that this is a linear constraint of $K-2$ parameters.
Starting values were usually successfully chosen to be those obtained from the marginal regression models in (\ref{eq:4}) while starting the association parameter in zero.

The Hessian of the procedure is the negative of a numerical version of the observed information $i(\psi) = -\frac{\partial^2 l (\psi)}{\partial \psi \partial \psi^T}$. As the asymptotic covariance matrix of the estimates is the inverse of the Fisher information $\expect{i(\psi)}$, asymptotic confidence intervals may be based on the numerical observed information, an approach taken in the example sections below.

\section{Association measures}
\label{sec:association-measures}

In Section \ref{sec:correlation-binary-y} we compare the bivariate logistic model arising from the Gumbel type 2 distribution to that from the AMH distribution by comparing the implied product moment correlation as also considered by \cite{heise_optimal_1996}. Below we additionally consider another measure of association between $X$ and $Y$, the odds ratio $\psi$.

\subsection{Correlation for binary $Y$}
\label{sec:correlation-binary-y}

Suppose that $Y \in \setdef{0,1}$. In the following we make a few comparisons between the bivariate logistic regression model obtained from latent variables following an AMH distribution and that obtained from the Gumbel type 2 model. We suppress the dependencies on covariates which may be easily reintroduced by substituting $\theta- Z_1 \beta_x$ for $\theta$ and $\tau - Z_2 \beta_y$ for $\tau$ in the formulas below.

Cf. Table \ref{tab:amh-moments} we may write the correlation as
\begin{equation}
  \label{eq:70}
  \begin{aligned}
    \rho_{\text{AMH}}(\theta,\tau ; \omega) &= \correl{X}{Y} \\
    &=
    \frac{
      \omega
    }{
      \left(
        e^{\theta/2}+e^{-\theta/2}
      \right)
      \left(
        e^{\tau/2} + e^{-\tau/2}
      \right)
      - \omega e^{-\theta/2-\tau/2}
    } .
\end{aligned}
\end{equation}

In \cite{heise_optimal_1996} a binary logistic regression is studied based on the Gumbel type 2 distribution in (\ref{eq:6}). This leads to the following correlation function,
\begin{equation}
  \label{eq:72}
  \rho_{\text{Type 2}}(\theta,\tau ; \omega) =
  \frac{\omega}{ \left( e^{\theta/2} + e^{-\theta/2} \right) \left( e^{\tau/2} + e^{-\tau/2} \right) } .
\end{equation}
The two correlation functions are plotted in Figure \ref{fig:correlations}.

\begin{figure}[htb]
  \centering
  \includegraphics[width=0.9\linewidth]{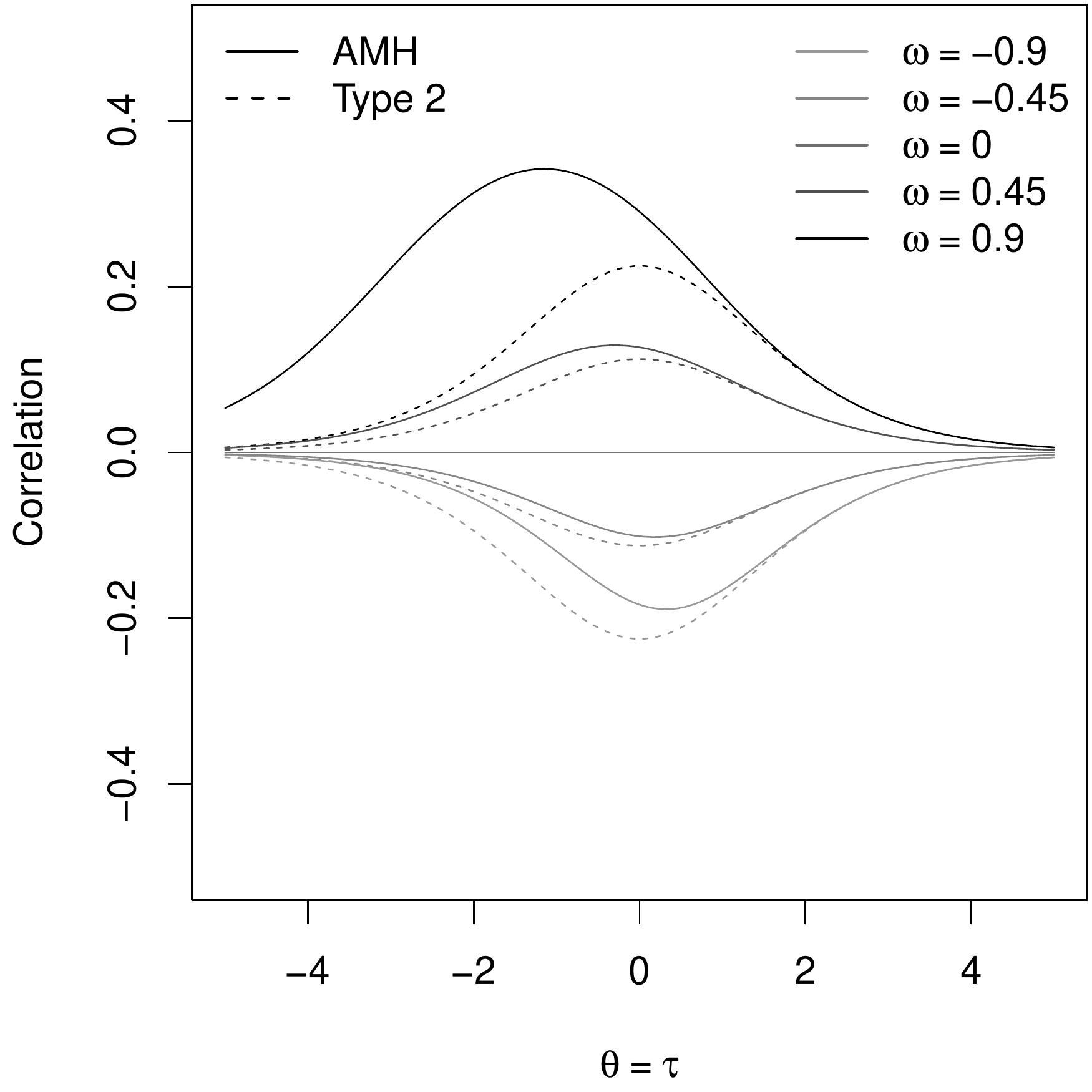}
  \caption{Plot of the correlation functions for the Gumbel type 2 (\ref{eq:72}) and the AMH logistic (\ref{eq:70}) model in the binary case ($K=2$). In both formulas $\theta=\tau$ is varied for $\omega=-0.9,-0.45,0,0.45,0.9$.}
  \label{fig:correlations}
\end{figure}

Viewing the correlations as functions on $\mathbb{R}^2$ for a fixed parameter \mbox{$\omega \in [-1;1]$}, it is simple to verify that,
\begin{equation}
  \label{eq:14}
  \rho_{\text{AMH}}  > \rho_{\text{Type 2}} .
\end{equation}
Moreover, by differentiation of $\rho_{\text{AMH}}$ one may further check that the AMH correlation is convex/concave for $\omega$ negative/positive attaining its global extreme when $\theta = \tau = \frac{1}{2} \log(1-\omega)$. The value of $\rho_{\text{AMH}}$ at this extremum is
\begin{equation}
  \label{eq:25}
      \frac{
      \omega
    }{
      2 
    }
    \cdot
    \frac{1}{1  + \sqrt{1-\omega}} .
\end{equation}
Particularly, we obtain the correlation bounds $- 1/\left(2 \left[1 + \sqrt{2}\right]\right)$ and $1/2$. In comparison, the correlation $\rho_{\text{Type 2}}$ attains its extremum $\omega/4$ when $\theta=\tau=0$, and is thus bounded between $-1/4$ and $1/4$.

It is worth noting here that when dealing with a bivariate distribution with Bernoulli marginals and marginal success probabilities, say, $\pi_1$  and $\pi_2$, respectively, the correlation cannot generally for any given marginals reach the limits $\pm 1$, the specific bounds depending on the marginals. Indeed, if $\pi_{11}$ is the joint success probability this must satisfy the Fr\'{e}chet-Hoeffding bounds in order for the bivariate distribution to be well-defined \citep[][Thm. 2.2.3]{nelsen_introduction_2006}. This implies that,
\begin{equation}
  \label{eq:34}
  \begin{aligned}
    &\max
    \left(
      - \sqrt{
        \frac{
          {\pi_1 \pi_2}
        }{
        \left(1 - \pi_1 \right) \left(1 - \pi_2 \right)
      }
      }
      ,
      -
      \sqrt{
      \frac{
        {\left(1-\pi_1 \right) \left(1-\pi_2 \right)}
      }{
        {\pi_1  \pi_2}
      }
      }
    \right) \\
    &\quad
    \leq
    \frac{\pi_{11} - \pi_1 \pi_2}{\sqrt{\pi_1 \left(1-\pi_1\right) \pi_2 \left(1-\pi_2\right)}} \\
    &\qquad
    \leq
    \min \left(
      \sqrt{
        \frac{
          \pi_1 \left(1 - \pi_2\right)
        }{
          \left(1 - \pi_1 \right) \pi_2
        }
      }
      ,
      \sqrt{
        \frac{
          \left(1 - \pi_1\right) \pi_2
        }{
          \pi_1 \left(1 - \pi_2 \right) 
        }
      }
    \right) ,
  \end{aligned}
\end{equation}
meaning that the correlation cannot attain the limits $\pm 1$ for arbitrary marginals. However, when $\pi_1 = \pi_2$ the upper bound is $1$, and if further $\pi_1 = \pi_2 = 1/2$ the lower bound is $-1$.

\cite{heise_optimal_1996} study a bivariate logistic regression model based on the Gumbel type 2 distribution in (\ref{eq:6}), and propose to let the bounds on the association parameter depend on the marginals, an approach also taken in \cite{li_two-dimensional_2011}. This is based on the observation that it is only necessary for the implied cell probabilities $p(x,y)$ to lie between zero and one for the bivariate binary distribution to be well defined. This will in particular be the case when $H$ is a bivariate distribution owing to the Fr\'{e}chet-Hoeffding bounds, but will also be satisfied for a wider span of association parameters. The same idea could be applied to the model under present study. This method has the benefit that for example the correlation can take on a wider span of values, while downsides include a more complicated parameter space, and that the latent distribution is no longer a proper probability distribution, which may be a hindrance to interpretation in some applications.

\subsection{The odds ratio}
\label{sec:odds-ratio}

The global (or cumulative) odds ratio $\psi_k$ between $X$ and $Y$ at the $k$'th scale level may be seen as the odds ratio in a table which is formed by collapsing counts over the cells $Y \leq k$ and $Y > k$,
\begin{equation}
  \label{eq:10}
  \begin{aligned}
    \psi_k &=
    \frac{
      \cproba{Y > k}{X = 1} / \cproba{Y \leq k}{X = 1}
    }{
      \cproba{Y > k}{X = 0} / \cproba{Y \leq k}{X = 0}
    } \\
    &=
    \frac{
      H(\theta,\tau_k) \left[1 - F(\theta) - F(\tau_k) + H(\theta,\tau_k) \right]
    }{
      \left[F(\theta) - H(\theta,\tau_k) \right] \left[F(\tau_k) - H(\theta,\tau_k) \right]
    } ,
  \end{aligned}
\end{equation}
where $F$ is the logistic function and $H$ is the bivariate distribution from (\ref{eq:9}). Thus, $\psi_k$ is the odds of $Y$ being larger than $k$ given that $X$ is one compared to the odds of $Y$ being larger than $k$ given that $X$ is zero. Note that this has a conditional interpretation given possible covariates. A bit of algebra shows that, now making explicit the dependence on the covariates,
\begin{equation}
  \label{eq:11}
  \psi_k =
  \frac{
    1 - \omega^2 F \left(\theta - Z_1 \beta_x - \log(1-\omega) \right) F \left(\tau_k - Z_2 \beta_y - \log(1-\omega) \right)
    }{
      1-\omega
    } ,
\end{equation}
thus expressing the odds ratio in terms of logistic functions. When $\omega=1$ this should be understood in a limiting sense, in which case $\psi_k = 2 + e^ {-(\theta - Z_1 \beta_x)} + e^ {-(\tau_k - Z_2 \beta_y)}$. We note that $\psi_k = 1$ if $\omega = 0$ and, for a fixed $\omega \not = 0$, the odds ratio $\psi_k$ varies between $1+\omega$ and $1/(1-\omega)$. When $\omega$ is allowed also to vary in $[-1;1]$, the odds ratio will vary in the positive reals. Conversely, if fixing the thresholds $\theta$ and $\tau_k$ as well as the covariates, the odds ratio varies in
\begin{equation}
  \label{eq:33}
  \left[
    \frac{1 - F(\theta - Z_1 \beta_x -\log(2)) F(\tau_k - Z_2 \beta_y - \log(2))}{2}
    ;
    2 + e^{-(\theta - Z_1 \beta_x)} + e^{-(\tau_k - Z_2 \beta_y)}
  \right],
\end{equation}
for $\omega \in [-1;1]$.
Further, (\ref{eq:11}) implies an approximation for the log odds ratio. Linearising the functions $z \mapsto \log(1-z)$ and $z \mapsto e^{z}-1$ we obtain the following approximation,
\begin{equation}
  \label{eq:12}
  \begin{aligned}
    \log \psi_k &\approx \left[\log(1+\omega) - \log \left(\frac{1}{1-\omega} \right) \right]
    F \left(\theta - Z_1 \beta_x - \log(1-\omega) \right)
    F \left(\tau_k - Z_2 \beta_y - \log(1-\omega) \right) \\
    &\quad+
    \log \left(\frac{1}{1-\omega} \right) ,
  \end{aligned}
\end{equation}
with error $o \left(\omega^2 F \left(\theta - Z_1 \beta_x - \log(1-\omega) \right) F \left(\tau_k - Z_2 \beta_y - \log(1-\omega) \right) \right)$, and we recognise a four parameter logistic function when fixing either of the marginals. The dependence of the odds ratio on the covariates in $Z_1$ and $Z_2$ may be interpreted as the covariates modifying the effect of $X$ on $Y$. The approximation in (\ref{eq:12}) shows that the effect modification on the log odds ratio scale essentially has an logistic shape. Further, the effect modification is bounded by $\omega$, the degree of association between $X$ and $Y$, in the sense that it is contained in the interval given by the asymptotes of the four-parameter logistic function in (\ref{eq:12}), $\left[\log(1+\omega) ; \log \left(\frac{1}{1-\omega}\right) \right]$. Thus, the larger the association between $X$ and $Y$ as measured by $\omega$, the larger the possible effect modification of a covariate. The maximal span of the effect modification, or the largest observable odds ratio ratio between any set of covariates, is $\frac{1}{1-\omega^2}$, which may be taken as an alternative interpretation of $\omega$.

\subsection{Modelling the association parameter}
\label{sec:modell-assoc-param}

As follows from the discussion above, the bivariate logistic regression model based on the AMH distribution implies that common association measures are modified by the covariates from the marginal models. For example, on the log odds ratio scale, the effect modification has an approximately logistic shape as follows from (\ref{eq:12}). This may prove too restrictive for some applications in which case it is desirable to model the association parameter. For numerical purposes this may be done on the Fisher $\zeta$-scale, and as there is no straightforward interpretation of the coefficients it is advisable to allow a flexible form of the model, i.e. modelling categorical variables as factors and continuous variables using for example splines.

\section{Example analyses}
\label{sec:example-analys-cogn}

Below we illustrate the bivariate logistic regression model based on the AMH distribution presented above by applying it to two data sets, the first concerned with trekking patterns of Norwegian hikers, while the second is a cognitive experiment in which subjects are asked to perform a visual task and rate their awareness.

The analyses were performed by implementing the model in \textsf{R} (version 3.4.4). An \textsf{R}-package containing the functions along with documentation may be obtained as an archive of source files by contacting the corresponding author.

\subsection{Norwegian trekking data}
\label{sec:norwegian-tour-data}

In \cite{haakenstad_skogbehandling_1972}, the authors survey various aspects of silviculture in Norway including a response to questionnaires delivered by post in Oslo from 365 hikers who were asked about their trekking habits, in particular how often they hike during the summer season and the length of a typical hike. Table \ref{tab:norwegian-tour-data} is a tabulation with slightly coarser groupings than the original data .

\begin{table}[htb]
  \centering
  \begin{tabular}{l|ccccc|r}
    & $<$2.5 & 2.5-5 & 5-10 & 10-20 & $>$20 & Sum \\ 
    \hline
    Rarer &  33 &  45 &  60 &  26 &   6 & 170 \\ 
    Weekly &  14 &  29 &  80 &  56 &  16 & 195 \\ 
    \hline
    Sum &  47 &  74 & 140 &  82 &  22 & 365 
  \end{tabular}
  \caption{Observed frequencies in the Norwegian trekking data. The rows indicate how often a respondent treks during a season, weekly or rarer, while the columns indicate the length of a typical hike in kilometres.}
  \label{tab:norwegian-tour-data}
\end{table}

As only aggregate level data is available with no covariates we consider the following intercept-only model. Denote the hiking frequency by $X \in \setdef{\text{"Rarer''}, \text{"Weekly''}}$ and the length of the hike by $Y \in \{<\text{2.5}, \text{2.5-5}, \text{5-10}, \text{10-20},$ $>\text{20}\}$, an ordinal variable with $K=5$ levels, which we number by $k=1,\ldots,5$. We observe pairs $(x_i, y_i)$ for $i=1,\ldots, 365$. The model for one such observation (dropping the $i$ subscript) is
\begin{equation}
  \label{eq:35}
  \begin{aligned}
    M:
    \begin{cases}
      \setdef{X = \text{"Weekly''}} = \setdef{X^{\ast} > \theta} \\
      \setdef{Y \leq k} = \setdef{Y^{\ast} \leq \tau_k} \\
      (X^{\ast}, Y^{\ast}) \sim \Lambda_2^{\ast} (\omega) .
    \end{cases}
  \end{aligned}
\end{equation}
Estimation is performed by maximizing the likelihood in (\ref{eq:129}) and the resulting estimates are given in Table \ref{tab:tour-data-estimates} with confidence intervals. Note that $\omega$ is estimated on the Fisher $\zeta$-scale and then backtransformed.

\begin{table}[ht]
  \centering
  \begin{tabular}{l rrr}
    & Estimate & 2.5\% & 97.5\% \\
    \hline
    \rule{0pt}{1em}$\hat{\theta}$ & -0.14 & -0.34 & 0.07 \\
    \hline
    $\hat{\tau}_1$ & -1.92 & -2.22 & -1.61 \\ 
    $\hat{\tau}_2$ & -0.71 & -0.92 & -0.49 \\ 
    $\hat{\tau}_3$ & 0.92 & 0.69 & 1.15 \\ 
    $\hat{\tau}_4$ & 2.75 & 2.31 & 3.18 \\
    \hline
    $\hat{\omega}$ & 0.76 & 0.49 & 0.89
  \end{tabular}
  \caption{Parameter estimates for the model $M$ (\ref{eq:35}).}
  \label{tab:tour-data-estimates}
\end{table}

We observe for example that $\hat{\theta} = -0.14$ (\confint{-0.34}{0.07}), indicating that while there is a slight overweight of respondents that trek on a weekly basis, there is not sufficient evidence to reject the hypothesis $\theta = 0$, that there is an equal prevalence of hikers who trek on a weekly basis and those who trek less frequently. The estimate of $\hat{\omega} = 0.76$ (\confint{0.49}{0.89}) indicates a positive association between trek frequency and length, as we return to below.

Predicted counts from the model $M$ were obtained by multiplying the predicted probabilities by the total count $N=365$ and are given in Table \ref{tab:norwegian-tour-predicted}. We observe a good concordance between the counts predicted by the model and those actually observed. 

\begin{table}[ht]
  \centering
  \begin{tabular}{l|ccccc}
    & $<$2.5 & 2.5-5 & 5-10 & 10-20 & $>$20 \\ 
    \hline
    Rarer & 33.59 & 43.30 & 60.30 & 26.38 & 6.26 \\ 
    Weekly & 13.18 & 30.48 & 80.03 & 55.73 & 15.75 \\ 
  \end{tabular}
  \caption{Predicted cell counts from model $M$ (\ref{eq:35}). A goodness of fit statistic comparing the predicted to the observed counts in Table \ref{tab:norwegian-tour-data} yields $\chi^2 = 0.22$.} 
  \label{tab:norwegian-tour-predicted}
\end{table}

As already noted, the estimate of $\omega$ indicated a positive association between hike frequency and length. To further quantify this association, odds ratios $\psi_k$ were calculated for $k=1, 2, 3, 4$. Note that $\psi_k$ is the odds of trekking more than what corresponds to the $k$'th distance category given that one hikes on a weekly basis against the odds of doing so given that one hikes less frequently than once a week. The observed and estimated odds ratios are given in Table \ref{tab:norwegian-tour-OR}. Confidence intervals were calculated on log-scale using the delta-method and then backtransformed.

\begin{table}[ht]
  \centering
  \begin{tabular}{l|r|rrr}
    k & Observed OR & Predicted OR & 2.5\% & 97.5\% \\ 
    \hline
    1 & 3.11 & 3.40 & 1.98 & 5.86 \\ 
    2 & 3.00 & 2.87 & 1.96 & 4.21 \\ 
    3 & 2.52 & 2.43 & 1.85 & 3.19 \\ 
    4 & 2.44 & 2.30 & 1.80 & 2.93 
  \end{tabular}
  \caption{Observed and predicted odds ratios for $k=1,2,3,4$. Confidence intervals were calculated on log scale using the delta-method and then backtransformed.} 
  \label{tab:norwegian-tour-OR}
\end{table}

For example, we estimate that the odds of trekking more than 10 kilometres is $143 \%$ (\confint{85 \%}{219 \%}) higher among those who hike on a weekly basis compared to those who hike less frequently.

Finally, as an example of an association measure of interest related to the latent distribution we estimate the expected wear-and-tear on the hiking trails as measured by $\expect{X^{\ast} Y^{\ast}}$, i.e. how many kilometres are the trails subject to in a season from an average hiker. As it follows from Table \ref{tab:amh-moments}, this quantity may be estimated up to an arbitrary precision as a polynomial in $\omega$ of a certain degree, and an asymptotic confidence interval may be determined from the delta-method. The quantity is given in units of standard deviations of the underlying variables. We estimate the quantity using the first ten terms of the series thus expecting accuracy up to two decimals,
\begin{equation}
  \label{eq:36}
  \begin{aligned}
  \expect{X^{\ast} Y^{\ast}} &\longleftarrow \sum_{n=1}^{10} \frac{\hat{\omega}^n}{n^2} \\
  &=
  0.99, 95 \% \text{CI} = [0.55 ; 1.43] .
\end{aligned}
\end{equation}
More specifically, suppose that one were willing to think of $X^{\ast}$ as the standardised number of hikes in a season which might have standard deviation $\sigma_x$ and that $Y^{\ast}$ were the length of the typical hike with standard deviation $\sigma_y$ in kilometres. To obtain a rough guess for the standard deviations, we compared the estimated thresholds to the mid-points of the intervals (for the tour frequency $X$ this was done marginally in a table with three levels of the frequency). This led us to tentatively take $\sigma_x = 5.4$ number of hikes per season and $\sigma_y = 2.9$ kilometres, and thus we estimate that the average hiker subjects the trial to $15.7$ kilometres (\confint{8.7}{22.7}) in the summer season.

\subsection{A cognitive experiment}
\label{sec:cognitive-experiment}

In a cognitive experiment performed at Aarhus University Hospital, researchers collected data on 20 participants who were asked to perform a series of visual recognition tasks. Specifically, in each trial the subject must report whether the letter presented on a monitor was a lower-case ``h'' or ``b'', the image being presented for a varying duration of 16, 33 or 100 milliseconds, and additionally use a perception awareness scale (PAS) with four levels to rate awareness of the given stimulus information. The experiment is of a so-called inclusion-exclusion type, where the task objective is varied so that a participant in an inclusion trial must report what is seen, while in an exclusion trial reporting the opposite. Each experimental type (in-/exclusion) is investigated over the three settings of stimulus duration in 80 replications per duration for the 20 participants leading to a total of $2 \cdot 3 \cdot 80 \cdot 20 = 9600$ measurements. For the following illustration we will focus on the $4800$ measurements of the inclusion task.

\subsubsection{Inclusion task}
\label{sec:exclusion-task}

Let $(X_{it}, Y_{it})$ be the accuracy and PAS rating of subject $i=1,\ldots, 20$ in the $t=1, \ldots, 240$'th trial. We initially pose the following model, in which, unrealistically, observations on the same individual are assumed to be independent.
\begin{equation}
  \label{eq:37}
  \begin{aligned}
    M_1:
    \begin{cases}
      \setdef{X_{it} = \text{"Correct''}} = \setdef{X_{it}^{\ast} > \theta} \\
      \setdef{Y_{it} \leq k} = \setdef{Y_{it}^{\ast} \leq \tau_k}, \quad k=1, 2, 3 \\
      (X_{it}^{\ast}, Y_{it}^{\ast}) \sim
      \Lambda_2^{\ast} (\beta_x \cdot s_{it}, \beta_y \cdot s_{it}; \omega) ,
    \end{cases}
  \end{aligned}
\end{equation}
where $s_{it}=16, 33, 100,$ is the stimulus duration in milliseconds for the $t$'th trial of subject $i$.

We additionally fitted the following model allowing the association parameter to depend on the stimulus intensity.
\begin{equation}
  \label{eq:26}
  \begin{aligned}
    M_2:
    \begin{cases}
      \setdef{X_{it} = \text{"Correct''}} = \setdef{X_{it}^{\ast} > \theta} \\
      \setdef{Y_{it} \leq k} = \setdef{Y_{it}^{\ast} \leq \tau_k}, \quad k=1, 2, 3 \\
      (X_{it}^{\ast}, Y_{it}^{\ast}) \sim
      \Lambda_2^{\ast} (\beta_x \cdot s_{it}, \beta_y \cdot s_{it} ; \omega_{s_{it}}) .
    \end{cases}
  \end{aligned}
\end{equation}

The parameters in $M_1$ and $M_2$ were estimated by maximum likelihood and the resulting estimates with $95 \%$ confidence intervals are given in Table \ref{tab:estimates-M-M_om}. Estimates of thresholds and $\beta_x$ and $\beta_y$ are very similar between the two models, while there is a clear dependence of the association parameter on the stimulus duration, particularly it is lower at 16 milliseconds compared to intermediate and high duration. 

\begin{table}[htb]
  \centering
  \begin{tabular}{l | ccc | ccc}
    &\multicolumn{3}{c}{$M_1$} & \multicolumn{3}{c}{$M_2$} \\
    & Estimates & 2.5 \% & 97.5 \% & Estimates & 2.5 \% & 97.5 \% \\ 
    \hline
    \rule{0pt}{1.1em}$\hat{\theta}$ & -0.2955 & -0.4280 & -0.1630 & -0.3125 & -0.4440 & -0.1810 \\
    \hline
    $\hat{\tau}_1$ & 0.1176 & 0.0212 & 0.2140 & 0.1127 & 0.0170 & 0.2084 \\ 
    $\hat{\tau}_2$ & 1.9491 & 1.8426 & 2.0556 & 1.9511 & 1.8450 & 2.0572 \\ 
    $\hat{\tau}_3$ & 3.7595 & 3.6024 & 3.9167 & 3.7571 & 3.6006 & 3.9136 \\
    \hline
    \rule{0pt}{1em}$\hat{\beta}_x$ & 0.0337 & 0.0300 & 0.0375 & 0.0335 & 0.0297 & 0.0372 \\ 
    $\hat{\beta}_y$ & 0.0429 & 0.0410 & 0.0449 & 0.0430 & 0.0410 & 0.0449 \\
    \hline
    $\hat{\omega}$ & 0.8846 & 0.8257 & 0.9245 &&&\\
    $\hat{\omega}_{16}$ &&&& 0.6648 & 0.5455 & 0.7577 \\ 
    $\hat{\omega}_{33}$ &&&& 0.9300 & 0.8404 & 0.9701 \\ 
    $\hat{\omega}_{100}$ &&&& 0.9581 & 0.9096 & 0.9808 \\ 
  \end{tabular}
  \caption{Parameter estimates with $95 \%$ confidence intervals for $M_1$ and $M_2$. In $M_2$ the association parameter is allowed to depend on the stimulus intensity $s \in \setdef{16, 33, 100}$.}
  \label{tab:estimates-M-M_om}
\end{table}

To evaluate the association between task accuracy and awareness rating, odds ratios $\psi_k$ were estimated for $k=1, 2, 3$. Note that $\psi_k$ is the odds of rating one's awareness as higher than $k$ given that one has given a correct response against the odds of rating awareness higher than $k$ having given an incorrect response. $\psi_k$ may thus be interpreted as a measure of metacognition as it quantifies the participants' ability to separate correct from incorrect reports. The resulting estimates with $95 \%$ delta method confidence intervals are shown in Figure \ref{fig:obs-pred-OR-mod2_omega} along with the observed odds ratios.

\begin{figure}[htb]
  \centering
  \includegraphics[width=0.99\linewidth]{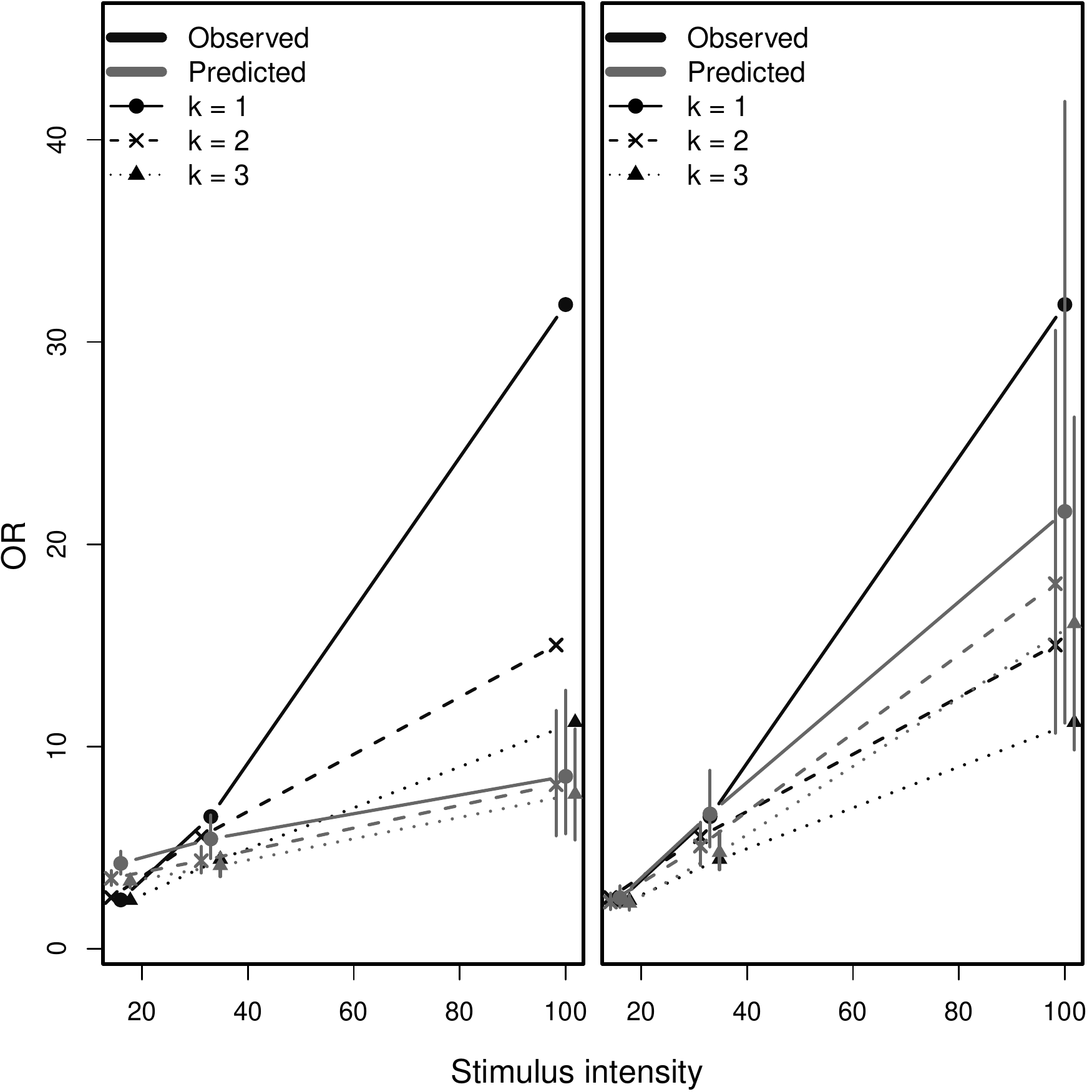}
  \caption{Observed and predicted odds ratios $\psi_k$ from $M_1$ (left) and $M_2$ (right) for $k=1,2,3$. In the left-hand plot there is no covariates in the $\omega$-model, while the right-hand plot depicts the odds ratios under a model with a separate $\omega$ for each stimulus intensity. Values are jittered on the x-axis.}
  \label{fig:obs-pred-OR-mod2_omega}
\end{figure}

For the model $M_1$ in the left-hand side of the figure, we observe that the predicted odds ratios are quite different from those observed, evidently overestimating the association at the lowest stimulus intensity, while underestimating the odds ratios at the high intensity. The odds ratios predicted by model $M_2$ seem to capture the dependence much better by including the stimulus intensity in the association parameter. Note that the wide confidence intervals at $s=100$ especially for $k=1$ is due to a low prevalence of incorrect answers and low prevalence of small PAS ratings at the high stimulus intensity.

As a final model we attempt to accommodate the repeated measurements on the subject level by allowing for a subject-specific level of accuracy and awareness,
\begin{equation}
  \label{eq:29}
  \begin{aligned}
    M_3:
    \begin{cases}
      \setdef{X_{it} = \text{"Correct''}} = \setdef{X_{it}^{\ast} > \theta} \\
      \setdef{Y_{it} \leq k} = \setdef{Y_{it}^{\ast} \leq \tau_k}, \quad k=1, 2, 3 \\
      (X_{it}^{\ast}, Y_{it}^{\ast}) \sim
      \Lambda_2^{\ast} (\beta_x \cdot s_{it} - a_{x,i}, \beta_y \cdot s_{it} - a_{y,i} ; \omega) ,
    \end{cases}
  \end{aligned}
\end{equation}
where we impose the restrictions $\sum_{i=1}^{N} a_{x,i} = \sum_{i=1}^{N} a_{y,i} = 0$ for identifiability. Below we discuss the random effect version of $M_3$, which is usually to be preferred. The sign for the subject-specific level is chosen so that $\theta+a_{x,i}$ and $\tau_k+a_{y,i}$ may be interpreted as subject specific intercepts. We do not give the estimates for $M_3$, but Figure \ref{fig:obs-pred-curves-m3} depicts the observed and predicted mean accuracy and mean awareness over stimulus intensity for four participants in the experiment.

\begin{figure}[htb]
  \centering
  \includegraphics[width = 0.9\linewidth]{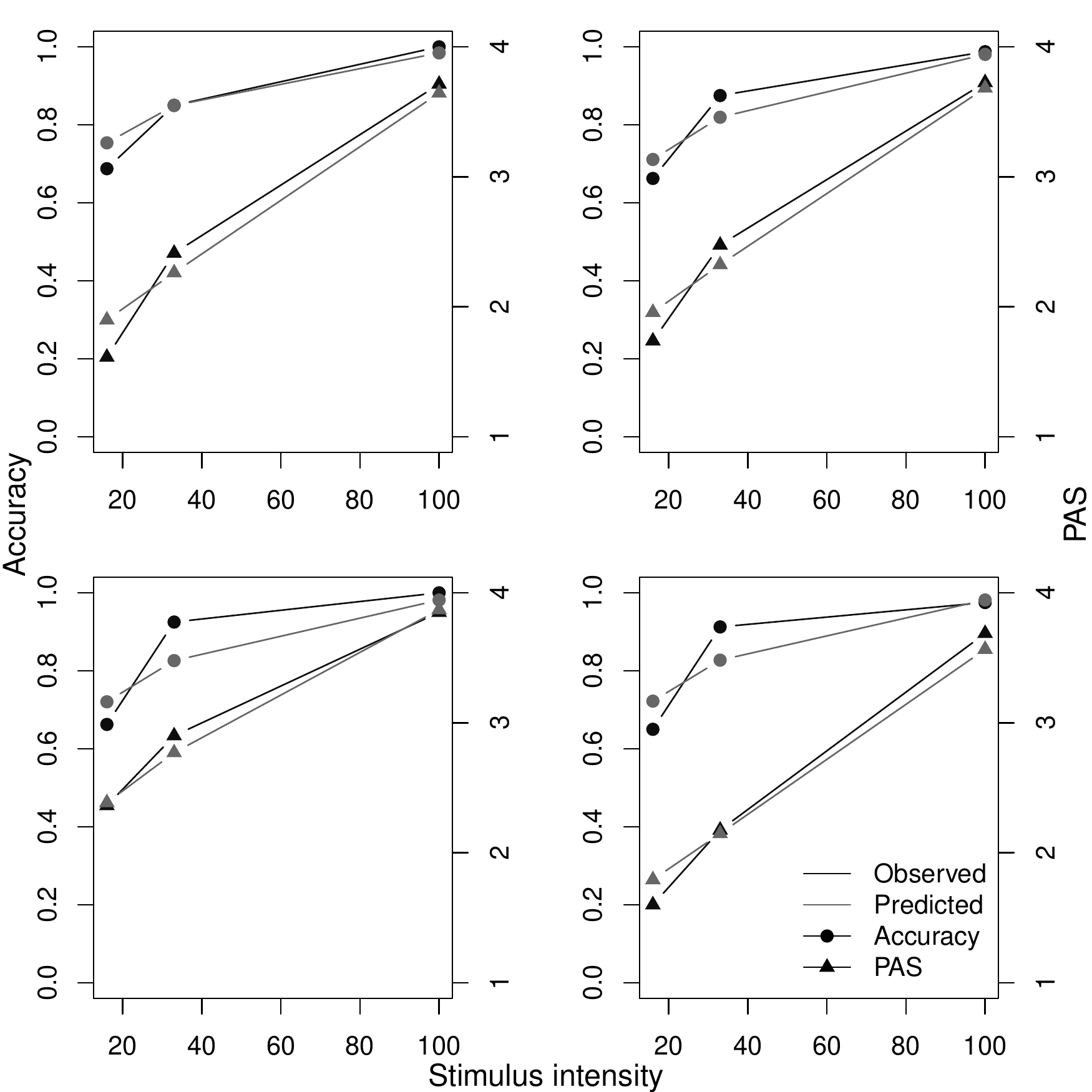}
  \caption{Observed and predicted accuracy and PAS from $M_3$ for four participants.}
  \label{fig:obs-pred-curves-m3}
\end{figure}

The interpretation of latent variables as representing underlying information is widespread in psychological models for cognitive experiments, where it is known as Signal Detection Theory \citep{green_signal_1966}, extensions of which exist also for metacognition \citep[see for example][]{kristensen_regression_forth}. Here we give such an interpretation based on Proposition \ref{sec:prop-amh-mixture}. As follows from (\ref{eq:28}) we may see the latent evidence used for visual identification and rating awareness as arising as a maximin principle applied to multiple underlying evidence processes, i.e. the identification evidence $X^{\ast}$ is an attempt to maximise the minimum information for a number of replicates of a number of processes. The association between the latent evidences $X^{\ast}$ and $Y^{\ast}$ comes from the fact that we maximise over the same number of processes, the expected number being $1/(1-\omega)$. From $M_2$ we estimate that an increase of stimulus intensity from low to intermediate increases the expected number of latent processes by a factor $(1-\hat{\omega}_{16}) / (1-\hat{\omega}_{33})  = 4.8$ (\confint{2.0}{11.7}). Similarly, comparing high stimulus intensity to low we estimate a relative ratio of number of latent processes as $8.0$ (\confint{3.5}{18.5}).

\section{Random effects}
\label{sec:random-effects}

In the longitudinal setting one generally needs to account for the correlation between replicate measurements on a subject and we may introduce random intercepts on a subject level through conditioning in the marginal models. Following the approach from Section \ref{sec:bivar-logist-models} we specify conditional random intercept models for the marginals,
\begin{equation}
  \label{eq:16}
  \begin{aligned}
    &\logit \cproba{X_{it}=1}{\alpha_{x,i}} = \theta + Z_{1, it} \beta_x + \alpha_{x,i},
    \quad \alpha_{x,i} \sim N(0, d_x^2) , \\
    &\logit \cproba{Y_{it} \leq k}{\alpha_{y,i}} = \tau_k - Z_{2,it} \beta_y + \alpha_{y,i},
    \quad \alpha_{y,i} \sim N(0, d_y^2).
  \end{aligned}
\end{equation}
and, again, stacking the latent variables,
\begin{equation}
  \label{eq:17}
  \begin{pmatrix}
    X^{\ast}_{it} \\
    Y^{\ast}_{it}
  \end{pmatrix} =
  \begin{pmatrix}
    Z_{1,it} & 0 \\
    0 & Z_{2,it}
  \end{pmatrix}
  \begin{pmatrix}
    \beta_x \\
    \beta_y
  \end{pmatrix} +
  \begin{pmatrix}
    \alpha_{x,i} \\
    \alpha_{y,i}
  \end{pmatrix} +
  \begin{pmatrix}
    \epsilon_{x,it} \\
    \epsilon_{y,it}
  \end{pmatrix} ,
\end{equation}
for $(\epsilon_x, \epsilon_y)$ independent of $(\alpha_x, \alpha_y)$ and where $\epsilon_x$ is still independent of $\epsilon_y$, both being standard logistic. One idea is then to specify the random intercepts as correlated,
\begin{equation}
  \label{eq:3}
  \begin{pmatrix}
    \alpha_{x,i} \\
    \alpha_{y,i}
  \end{pmatrix} \sim
  N_2 \left(
  \begin{Bmatrix}
    0\\0
  \end{Bmatrix}
  ,
  \begin{Bmatrix}
    d_x^2 & d_{xy} \\
    d_{xy} & d_y^2
  \end{Bmatrix}
  \right) ,
\end{equation}
and the resulting model is often termed a \emph{joint mixed model} or, when $\alpha_x=\alpha_y$, a \emph{shared parameter model} \citep{verbeke_analysis_2014}. The correlation between the random intercepts in (\ref{eq:3}) implies a bivariate model in a marginal sense, i.e. marginally $X$ and $Y$ will be correlated but conditioning on the random effects, $X$ and $Y$ are independent. This may be an appropriate assumption in some scenarios. For example, the shared parameter model is common in rater comparison studies, where two raters both evaluate a number of items. Here it seems reasonable to assume that the two raters' assessments would be independent given the information in an item. To the contrary, in a cognitive experiment such as that analysed above it may not be realistic to assume that accuracy and confidence should be independent given the participant’s overall level of performance and confidence. In such cases, it would be relevant to introduce correlation between $\epsilon_x$ and $\epsilon_y$, which can be done for example by specifying the joint distribution of $(\epsilon_x, \epsilon_y)$ as the AMH distribution thus yielding a random effect version of the model studied in the present paper.

More precisely, we reuse the threshold model from (\ref{eq:27}) and take the joint conditional distribution of the latent variables $X_{it}^{\ast}, Y_{it}^{\ast}$ given $(\alpha_{x,i}, \alpha_{y,i})$ to be $H\left(Z_1 \beta_x - \alpha_{x,i} , Z_2 \beta_y - \alpha_{y,i} \right)$. We denote by $l_{it} (\psi \vert \alpha_{x,i}, \alpha_{y,i})$ the conditional log-likelihood for a single observation $(X_{it}, Y_{it})$ from this model, which is of the form in (\ref{eq:129}) (with $N=1$). Estimation may be based on the marginal likelihood function given by,
\begin{equation}
  \label{eq:30}
  \begin{aligned}
    l \left( \psi, D \right)
    &=
    \sum_{i=1}^{N}
    \log \int_{\mathbb{R}} \int_{\mathbb{R}}
    e^{\sum_{t=1}^{T} l_{it} (\psi \vert u, v)} \phi \left(u,v ; D \right)
    \: du \: dv
  \end{aligned} ,
\end{equation}
where $D$ is the covariance matrix of the random effects in (\ref{eq:3}) and $\phi \left(\cdot, \cdot ; \Sigma \right)$ is the density of a bivariate normal distribution with zero mean vector and covariance matrix $\Sigma$.

The integral in (\ref{eq:30}) does not admit a closed form, but will need to be approximated. Common techniques \citep[see for example][for a review]{tuerlinckx_statistical_2006} include Monte Carlo integration, approximating the integrands for example by the Laplace method or quasi-likelihood as well as the use of deterministic integration methods such as Gauss-Hermite (GH) quadrature, a bivariate version of which we sketch in the following. Denote by $I_i$ the two-dimensional integral from (\ref{eq:30}) and write $L_i(u,v) = \exp({\sum_{t=1}^{T} l_{it} (\psi \vert u, v)})$ for the conditional likelihood as a function of the random effects. Univariate Gauss-Hermite quadrature would for some function $g$ and Hermitian weights  $w(x)=\exp(-x^2)$ approximate the integral $\int g(x) w(x) \: dx$  by the sum $\sum_{i=1}^R g(z_r) w_r$ where the $R$ weights $\setdef{w_r}_{r=1}^R$ and nodes $\setdef{z_r}_{r=1}^R$ may be derived from the Hermitian polynomials, or conveniently obtained from statistical software or tables.

This procedure may be extended to the bivariate case as follows. Supposing that $D$ is positive definite let $D=L L^T$ be its Cholesky decomposition where $L$ is a lower triangular matrix with strictly positive diagonal elements $l_1$ and $l_2$ and off-diagonal entry $l_{12} \in \mathbb{R}$. Note then that the random effect are distributed identically to $L$ times a two-dimensional vector of independent standard normal variables, so that we by transformation of the density obtain,
\begin{equation}
  \label{eq:31}
  \phi(u, v ; D) = \phi \left( \frac{1}{l_1} u \right)
  \phi \left( \frac{1}{l_2} v - \frac{l_{12}}{l_1 l_2} u \right) ,
\end{equation}
where the right-hand side is often termed the Cholesky parametrisation.  Inserting into the expression for the integral $I_i$ and performing substitution yields,
\begin{equation}
  \label{eq:32}
  \begin{aligned}
    I_i
    &=
    \int_{\mathbb{R}} \int_{\mathbb{R}}
    L_i (l_1 u, l_{12} u + l_2 v)
    \phi \left(u \right)
    \phi \left(v \right)
    \: du \:dv \\
    &=
    \frac{1}{\pi}
    \int_{\mathbb{R}} \int_{\mathbb{R}}
    L_i (l_1 \sqrt{2} u, l_{12} \sqrt{2} u + l_2 \sqrt{2} v)
    w \left(u \right)
    w \left(v \right)
    \: du \:dv .
  \end{aligned}
\end{equation}
Thus, we may simply apply the univariate GH quadrature method twice and approximate,
\begin{equation}
  \label{eq:38}
  I_i \approx
  \frac{1}{\pi} \sum_{i=1}^{R} \sum_{j=1}^{R}
  w_i w_j
  L_i (l_1 \sqrt{2} z_i, l_{12} \sqrt{2} z_i + l_2 \sqrt{2} z_j) ,
\end{equation}
for weights and nodes $\setdef{w_r, z_r}_{r=1}^R$ which may be obtained outside of the maximisation procedure as they are parameter independent. We thus replace the $N$ integrals in the log-likelihood in (\ref{eq:30}) by approximations of the form (\ref{eq:38}) and maximise the resulting likelihood as a function of $\psi$ and $(l_1, l_2, l_{12})$ to obtain estimates $\hat{\psi}$ and $\widehat{D} = \widehat{L} \widehat{L}^T$.

In the random effects model, association measures such as the odds ratio $\psi$ will depend on the random effects. In some cases, the subject-specific measure will be uninteresting and one may wish to report the population version of, say, the odds ratio given certain covariates. In this case random effects can be integrated from the probabilities using uni- and bivariate Gauss-Hermite quadrature as described above.

\section{Acknowledgements}
\label{sec:acknowledgements}

The authors would like to thank Dr Kristian Sandberg at the Center of Functionally Integrative Neuroscience (CFIN), Aarhus University and Aarhus University Hospital for allowing us to use the HB data set.

\medskip
 \renewcommand{\bibname}{References}
\small{\bibliography{bibliography}}

\appendix

\chapter{Proofs}
\label{sec:proofs}

\section{Proof of Proposition}
\label{sec:proof-proposition}

The cdf of the AMH distribution is given by
\begin{equation}
\label{eq:20}
  \begin{aligned}
    H(x,y;\omega) &=
    \frac{1}{1 + e^{-x} + e^{-y} + (1-\omega) e^{-x-y}} \\
    &=
    \frac{1-\omega}{1 - \omega + (1-\omega) \left( e^{-x} + e^{-y} + (1-\omega) e^{-x-y} \right)} \\
    &=
    \frac{
      1-\omega
    }{
      \left[1 + e^{-x + \log((1-\omega))} \right] \left[1 + e^{-y  +\log((1-\omega))} \right]  - \omega
    } \\
    &=
    -
    \frac{1-\omega}{\omega} \cdot
    \frac{
      1
    }{
      1 - 1/t
    } \\
    &=
    \frac{1-\omega}{\omega} \cdot
    \frac{
      t
    }{
      1 - t
    } ,
  \end{aligned}
\end{equation}
for $t=F(z_1) F(z_2) \omega$, where $F(z) = (1+e^{-z})^{-1}$ is the logistic function and
\begin{equation}
\label{eq:21}
  z_1 = x - \log(1-\omega), \quad z_2 = y - \log(1-\omega) .
\end{equation}
Thus, since $\absval{t} < 1$,
\begin{equation}
\label{eq:22}
  \begin{aligned}
    H(x,y;\omega) &=
    \frac{1-\omega}{\omega} 
    t
    \sum_{n=0}^{\infty} t^n \\
    &=
    \frac{1-\omega}{\omega} 
    \sum_{n=1}^{\infty} t^n \\
    &=
    \sum_{n=1}^{\infty} \omega^{n-1} (1-\omega) \left( F(z_1) F(z_2) \right)^n \\
    &= 
    \sum_{n=1}^{\infty} (1-\bar{\omega})^{n-1} \bar{\omega} \left( F(z_1) F(z_2) \right)^n ,
  \end{aligned}
\end{equation}
for $\bar{\omega}=1-\omega$. If we further assume $\omega \in (0,1)$, $\bar{\omega}$ may be interpreted as a probability and we obtain the formula,
\begin{equation}
  \label{eq:23}
  H(x,y;\omega) = \mathbb{E}_{N \sim \text{geom} (\bar{\omega})}
  \left[ \left( F(z_1) F(z_2) \right)^N  \right] ,
\end{equation}
which expresses the cdf as a mixture of logistic and geometric variables.

\end{document}